\documentclass[
    10pt,
    amsmath,
    amssymb,
    aps,
    nobibnotes,
    prl,
    superscriptaddress,
    twocolumn
]{revtex4-2}
\usepackage{amsthm}
\usepackage{bbm}
\usepackage[T1]{fontenc}
\usepackage{footmisc}
\usepackage{mathtools}
\usepackage{thmtools}
\usepackage{titlesec}
\usepackage[caption=false]{subfig}

\usepackage{silence}
\usepackage[
    bookmarksopen=true,
    bookmarksopenlevel=1,
    breaklinks=true,
    pdfencoding=auto,
    pdfpagemode=UseNone
]{hyperref}
\usepackage{xcolor}
\definecolor{medblue}{RGB}{0, 0, 100}
\definecolor{panblue}{RGB}{0,24,150}
\definecolor{carmine}{RGB}{150, 0, 24}
\hypersetup{colorlinks,
    linkcolor=carmine,
    citecolor=medblue,
    urlcolor=panblue
}

\newtheorem{theorem}{Theorem}
\newtheorem{corollary}[theorem]{Corollary}
\newtheorem{lemma}{Lemma}
\usepackage[capitalise]{cleveref}

\usepackage{tikz}
\usetikzlibrary{positioning}
\usetikzlibrary{decorations.pathmorphing}
\tikzset{~/.style={decorate, decoration={
    snake, amplitude=.5mm, segment length=3mm
}}}
\usepackage{pgfplots}
\pgfplotsset{compat=1.18}
\pgfplotsset{
    legend image with text/.style={
        legend image code/.code={%
            \node[anchor=center] at (0.3cm,0cm) {#1};
        }
    },
}

\makeatletter
\newcommand{\E}{\if@display\mathop{\mathbb{E}}\else\mathbb{E}\fi}
\makeatother

\let\oldmid\mid
\renewcommand{\mid}{\hspace*{-2pt}\oldmid\hspace*{-2pt}}

\DeclareMathOperator{\supp}{supp}
\DeclarePairedDelimiter{\abs}{\lvert}{\rvert}
\DeclarePairedDelimiter{\norm}{\lVert}{\rVert}
\DeclarePairedDelimiter{\ceil}{\lceil}{\rceil}
\DeclarePairedDelimiterX{\inner}[2]{\langle}{\rangle}{#1, #2}

\newcommand{\na}{\abs{A_i}}
\newcommand{\nx}{\abs{X_i}}

\titleformat{\section}[runin]{\itshape}{\appendixname~\thesection:~}{0pt}{}[---]
\titlespacing{\section}{\parindent}{\baselineskip}{0pt}

\newenvironment{proof*}{}{\hfill\qedsymbol}
\renewcommand{\qedsymbol}{$\blacksquare$}
\newtheoremstyle{prlthm}{}{}{}{\parindent}{\itshape}{\unskip---}{0.16667em}{\thmname{#1} \thmnumber{#2} \thmnote{#3}}
\theoremstyle{prlthm}
\renewcommand{\eqref}{\cref}

\AddToHook{cmd/appendix/before}{
    \onecolumngrid
    \begin{center}
    \textbf{\large End Matter}
    \end{center}
    \twocolumngrid
    \crefalias{section}{appendix}
    \setcounter{secnumdepth}{2} 
    \counterwithin{lemma}{section} 
    
}

\begin{document}

\title{Randomness Compression in Communication Networks}

\author{Yukari Uchibori}
\affiliation{Department of Physics and Astronomy, University of British Columbia, Vancouver, BC, Canada}

\author{Alice Zheng}
\thanks{Contact author: alicezheng@vt.edu}
\affiliation{Department of Computer Science, Virginia Polytechnic Institute and State University, Blacksburg, VA, USA}

\author{Anurag Anshu}
\affiliation{School of Engineering and Applied Sciences, Harvard University, Cambridge, MA, USA}

\author{Jamie Sikora}
\affiliation{Department of Computer Science, Virginia Polytechnic Institute and State University, Blacksburg, VA, USA}

\makeatletter
\hypersetup{
    pdftitle={\@title},
    pdfauthor={Yukari Uchibori, Alice Zheng, Anurag Anshu, Jamie Sikora}
}
\makeatother

\date{March 16, 2026}

\begin{abstract}
    Given a correlation generated by a (possibly quantum) communication network, we study the amount of shared randomness required to generate it.
    We develop a novel upper bound for approximating distributions generated by arbitrary networks and showcase instances where it significantly outperforms the best-known upper bounds for the exact case.
    This demonstrates that one can have substantial savings in resources if small perturbations are acceptable.
    We derive our bound using Hoeffding's inequality and apply it to various commonly-used communication networks such as the Bell scenario and triangle scenario.
\end{abstract}

\maketitle

\section{Introduction}

The phenomenon of quantum nonlocality reveals a multitude of peculiar properties of Nature.
From its theoretical conception~\cite{bell1964einstein} to the, now many, experimental realizations \cite{freedman1972experimental, fry1976experimental, aspect1982experimental, tittel1998violation, weihs1998violation, rowe2001experimental, moehring2004experimental} (including more recent loophole-free ones \cite{hensen2015loophole, shalm2015strong, giustina2015significant, storz2023loophole, drmota2025experimental}), nonlocality has ingrained itself as a conventional demonstration of practical quantum behavior \cite{brunner2014bell, scarani2019bell}.
It has relevance to various fields such as information theory \cite{cubitt2011zero}, communication complexity \cite{buhrman2010nonlocality}, and quantum cryptography \cite{sikora2014strong}, aiding the study of topics such as device-independent self-testing \cite{cirel1980quantum, mayers1998quantum} and dimension witnesses \cite{brunner2008testing}.
Determining the resources needed to perform certain tasks is crucial in practice and extends to many other areas of mathematics, computer science, and physics.
For example, given statistical data from a quantum experiment, it is possible to use nonlocality to devise device-independent tests to certify the necessity of entanglement~\cite{sikora2016device,wei2017device,wei2019device}.
On the other hand, we may seek to reproduce a given experiment in a smaller laboratory using fewer resources.
The latter may seem unlikely, as one can prove the existence of large Hilbert spaces~\cite{sikora2016minimum} and large entanglement~\cite{coladangelo2017all}.
In contrast, we demonstrate that in most tasks involving causal networks, large classical resources are not necessary to reproduce approximate probability distributions.

We illustrate in \cref{fig:examples} examples of networks with communication restrictions that have been widely studied.
These two networks are commonly used for various procedures studying phenomena such as nonconvexity \cite{donohue2015identifying}, nonlocality \cite{sunilkumar2025genuine}, and classical-quantum gaps \cite{da2025local}.

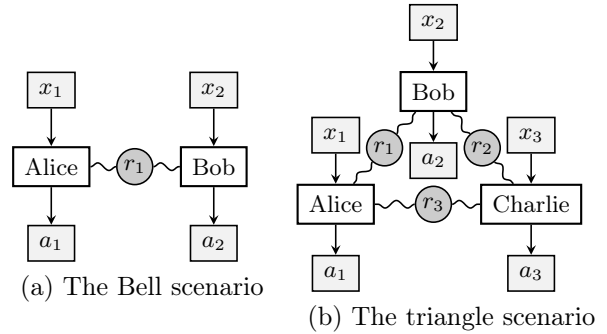
\begin{figure}[htp]
    \centering
    \subfloat[The Bell scenario]{%
        \begin{tikzpicture}[> = stealth, semithick, baseline]
    \tikzstyle{randomness}=[circle,draw=black,fill=black!20,inner sep=0.5mm]
    \tikzstyle{input}=[draw=black,fill=black!5,inner sep=1.5mm]
    \tikzstyle{output}=[draw=black,fill=black!5,inner sep=1.5mm]
    \tikzstyle{party}=[thick,draw=black,fill=white,inner sep=1.5mm]
    \node[party] (A) at (-1.11,0) {Alice};
    \node[party] (B) at (1.05,0) {Bob};
    \node[output,below=.5cm of A] (a1) {$a_1$};
    \node[output,below=.5cm of B] (a2) {$a_2$};
    \draw[->] (A) -- (a1);
    \draw[->] (B) -- (a2);
    \node[randomness] (r1) at (0,0) {$r_1$};
    \draw[~] (r1) -- (A);
    \draw[~] (r1) -- (B);
    \node[input,above=.5cm of A] (x1) {$x_1$};
    \node[input,above=.5cm of B] (x2) {$x_2$};
    \draw[->] (x1) -- (A);
    \draw[->] (x2) -- (B);
\end{tikzpicture}
        \label{fig:bell}
    }\hspace*{1em}%
    \subfloat[The triangle scenario]{%
        \begin{tikzpicture}[> = stealth, semithick, baseline]
    \tikzstyle{randomness}=[circle,draw=black,fill=black!20,inner sep=0.5mm]
    \tikzstyle{input}=[draw=black,fill=black!5,inner sep=1.5mm]
    \tikzstyle{output}=[draw=black,fill=black!5,inner sep=1.5mm]
    \tikzstyle{party}=[thick,draw=black,fill=white,inner sep=1.5mm]
    \pgfmathsetmacro{\phi}{210} 
    \pgfmathsetmacro{\rx}{1.5}
    \pgfmathsetmacro{\ry}{1}
    \node[party] (A) at ({\rx*cos(\phi)}, {\ry*sin(\phi)}) {Alice};
    \node[party] (B) at ({\rx*cos(\phi-120)}, {\ry*sin(\phi-120)}) {Bob};
    \node[party] (C) at ({\rx*cos(\phi-240)}, {\ry*sin(\phi-240)}) {Charlie};
    \node[output,below=.4cm of A] (a1) {$a_1$};
    \node[output,below=.4cm of B] (a2) {$a_2$};
    \node[output,below=.4cm of C] (a3) {$a_3$};
    \draw[->] (A) -- (a1);
    \draw[->] (B) -- (a2);
    \draw[->] (C) -- (a3);
    \node[randomness] (r1) at ({\rx*(cos(\phi)+cos(\phi-120))/2}, {\ry*(sin(\phi)+sin(\phi-120))/2}) {$r_1$};
    \node[randomness] (r2) at ({\rx*(cos(\phi-120)+cos(\phi-240))/2}, {\ry*(sin(\phi-120)+sin(\phi-240))/2}) {$r_2$};
    \node[randomness] (r3) at ({\rx*(cos(\phi-240)+cos(\phi))/2}, {\ry*(sin(\phi-240)+sin(\phi))/2}) {$r_3$};
    \draw[~] (r1) -- (A);
    \draw[~] (r1) -- (B);
    \draw[~] (r2) -- (B);
    \draw[~] (r2) -- (C);
    \draw[~] (r3) -- (C);
    \draw[~] (r3) -- (A);
    \node[input,above=.4cm of A] (x1) {$x_1$};
    \node[input,above=.4cm of B] (x2) {$x_2$};
    \node[input,above=.4cm of C] (x3) {$x_3$};
    \draw[->] (x1) -- (A);
    \draw[->] (x2) -- (B);
    \draw[->] (x3) -- (C);
\end{tikzpicture}
        \label{fig:triangle}
    }%
    \caption{%
        Examples of networks utilizing shared classical resources.
        The $x_i$s are inputs, the $a_i$s are outputs, and the $r_i$s are sources of shared randomness, each shared between a selected subset of the parties.
    }
    \label{fig:examples}
\end{figure}

The setup in \cref{fig:bell} is called the \emph{Bell scenario}.
It disallows communication between Alice and Bob but gives both access to a shared resource.
When this resource is a source of randomness, both parties receive an identical classical value $r_1$ that they may use to compute their outputs $a_1$ and $a_2$ from the respective inputs $x_1$ and $x_2$.
The distribution $p(a_1, a_2 \mid x_1, x_2)$ of this network is
\begin{equation}\label{eq:bellprob}
    \E_{r_1} \big[ \, p_A(a_1 \mid x_1, r_1) \cdot p_B(a_2 \mid x_2, r_1) \, \big]\ ,
\end{equation}
where $p_A$ and $p_B$ describe Alice and Bob's output strategies, which may or may not involve private randomness.

Another well-studied setting with no communication is the \emph{triangle scenario} depicted in \cref{fig:triangle}, which involves three parties and three sources of shared randomness, none of which is available to all three parties simultaneously.
In this setting, the outcome distribution $p(a_1, a_2, a_3 \mid x_1, x_2, x_3)$ is similar to \eqref{eq:bellprob}, being
\begin{equation}
    \E_{r_1, r_2, r_3} \big[ p(a \mid x, r_1, r_2, r_3) \big],
\end{equation}
where $r_1$, $r_2$, and $r_3$ are independent and $p(a \mid x, r_1, r_2, r_3)$ can be decomposed as%
\begin{equation}\label{eq:triangleprob}
    p_A(a_1 \mid x_1, r_1, r_3) \cdot p_B(a_2 \mid x_2, r_1, r_2) \cdot p_C(a_3 \mid x_3, r_2, r_3)
\end{equation}
with $p_C$ describing Charlie's output strategy.
It is important to note that without imposing the structure of \eqref{eq:triangleprob}, $p(a \mid x, r_1, r_2, r_3)$ could take the form of any conditional joint distribution, and only certain distributions can be generated from the triangle scenario in \cref{fig:triangle} (for instance, see \cite[Examples~1~and~2]{wolfe2019inflation}).

Despite the many ways to create networks, its defining properties are the cardinality of its inputs, outputs, shared resources (whether they are classical, quantum, or something else), and the nature of its communication.
When focusing on its classical resources, such as shared randomness, we can write the outcome probability distribution $p(a \mid x)$ in a general network as
\begin{equation} \label{eq:general_P_classical_form}
    \E_{r_1, \ldots, r_m}
    \big[ \, p(a \mid x, r_1, \ldots, r_m) \, \big]\ ,
\end{equation}
and the network may be viewed as in \cref{fig:general}.

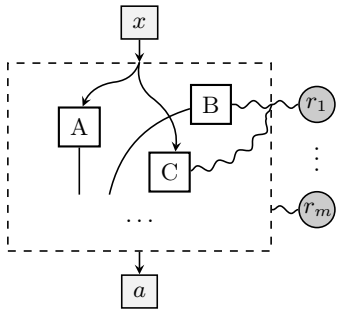
\begin{figure}[htp]
    \centering
    \begin{tikzpicture}[> = stealth, semithick]
    \tikzstyle{randomness}=[circle,draw=black,fill=black!20,inner sep=0.5mm]
    \tikzstyle{input}=[draw=black,fill=black!5,inner sep=1.5mm]
    \tikzstyle{output}=[draw=black,fill=black!5,inner sep=1.5mm]
    \tikzstyle{party}=[thick,draw=black,fill=white,inner sep=1.5mm]
    \node[rectangle,dashed,draw=black,minimum height=2.5cm,minimum width=3.5cm] (network) at (0, 0) {};
    \node[output,below=.3cm of network] (a) {$a$};
    \draw[->] (network.south) -- (a);
    \node[party] (A) at (-.8, .4) {A};
    \node[party] (B) at (0.95, .7) {B};
    \node[party] (C) at (.4, -.2) {C};
    \node at (0.03, -.85) {$\dots$};
    \draw[->] (network.north) to[bend left] (-.31, 1) to[bend right] (A);
    \draw[->] (network.north) to[bend right] (.15, .8) to[bend left] (C);
    \draw (A) -- (-.8, -.5);
    \draw (B) to[bend right] (-.4, -.5);
    \coordinate[above=.7cm of network.east] (r1anchor);
    \coordinate[below=.7cm of network.east] (rmanchor);
    \node[randomness,right=.35cm of r1anchor] (r1) {$r_1$};
    \node[randomness,right=.35cm of rmanchor] (rm) {\rlap{\hspace*{-0.1em}$r_m$}\phantom{$r_1$}};
    \path (r1) -- (rm) node[midway,rotate=90] {$\dots$};
    \draw[~] (r1) -- (r1anchor);
    \draw[~] (rm) -- (rmanchor);
    \draw[~] (r1anchor) -- (B);
    \draw[~] (C) to[bend right] (r1anchor);
    \node[input,above=.3cm of network] (x) {$x$};
    \draw[->] (x) -- (network.north);
\end{tikzpicture}
    \caption{%
        General communication network with input $x \in X$, output $a \in A$, and $m$ sources of shared randomness $r_i \in R_i$.
        All other resources such as shared entanglement and classical and quantum communication are within the dotted box.
    }
    \label{fig:general}
\end{figure}

In this Letter, we provide a novel upper bound on the amount of shared randomness needed to approximate any distribution generated by a network of any given size.
To accomplish this for a network generating a distribution $p(a \mid x)$, we replace its randomness sources $R_i$ with $Q_i$ (often of smaller cardinalities) while ensuring that the newly-generated distribution $\hat{p}(a \mid x)$ is $\epsilon$-close to the original.
For an example, consider the Bell Scenario in \cref{fig:bell} with input sets $X_1$, $X_2$, output sets $A_1$, $A_2$, and randomness source $R$.
Then, any Bell distribution can be approximated via a randomness source $Q$ of cardinality
\begin{equation}\label{eq:bell}
    \abs{Q} = \ceil*{ \frac{1}{2 \epsilon^2}\ln(2 \abs{X_1} \abs{X_2} \abs{A_1} \abs{A_2}) }\ .
\end{equation}
For the triangle scenario in \cref{fig:triangle} with input sets $X_1$, $X_2$, $X_3$ and output sets $A_1$, $A_2$, $A_3$, any distribution $p(a_1, a_2, a_3 \mid x_1, x_2, x_3)$ can be approximated via three randomness sources $Q_i$, each of cardinality
\begin{equation}\label{eq:triangle}
    \abs{Q_i} = \ceil*{ \frac{9}{2 \epsilon^2}\ln(2 \abs{X_1} \abs{X_2} \abs{X_3} \abs{A_1} \abs{A_2} \abs{A_3}) }\ .
\end{equation}
Moreover, we show that for a general network as depicted in \cref{fig:general}, we can approximate $p(a \mid x)$ by replacing $m$ randomness sources $R_i$ with $m$ randomness sources $Q_i$, each of cardinality
\begin{equation}\label{eq:general}
    \abs{Q_i} = \ceil*{ \frac{m^2}{2 \epsilon^2} \ln(2 \abs{X} \abs{A}) }\ .
\end{equation}
In the case of independent inputs and outputs given to $h$ parties with each having size $\nx$ and $\na$, we have $\abs{X} = \nx^h$ and $\abs{A} = \na^h$, meaning that the right-hand side of \cref{eq:general} becomes
\begin{equation}\label{eq:lnxa}
    \ln(2\abs{X}\abs{A}) = \ln(2) + h \ln( \nx \na )\ .
\end{equation}
Notably, this implies our bound scales linearly in the number of parties $h$, whereas the best-known upper bound for the exact case is exponential in $h$.

\section{Our approach}

Denoting the randomness sources by random variables $R_i$, we study how much the number of outcomes of each $R_i$ may be reduced without greatly affecting the network.
The high-level idea is to replace each source $R_i$ with a new source $Q_i$ having a smaller set of possible values, without otherwise changing the network structure.
These $Q_i$ are created by randomly sampling from $R_i$ according to their probabilities, and assigning the probability of each outcome proportional to how many times it was sampled.
When $\abs{R_i}$ is much larger than the number of samples, many unlikely values from $R_i$ may never be chosen, thereby reducing the cardinality of $Q_i$.
As values with higher probabilities in $R_i$ tend to be correspondingly more frequently sampled by $Q_i$, the randomness sources bear a significant resemblance, and hence their exchange does not significantly impact the network distribution.

As we cannot rely on any specific network structure to obtain general bounds, we encapsulate the entire network as a single object with inputs and outputs, while preserving the inner workings.
This effectively treats all involved parties as one, and in practice makes the result applicable to any communication network, with any sort of additional resources involved.
This means we can compress shared randomness in networks containing shared entanglement and even resources that fall outside of the theory of quantum mechanics.

The bound itself is obtained through concentration inequalities, specifically through repeat applications of Hoeffding's inequality.
We start by selecting a single source of randomness $R$ to compress, isolating the rest of the network into the distribution $p(a \mid x, r)$ dependent on its values.
We apply Hoeffding's inequality on this distribution while substituting $R$ as described previously.
This gives us the following compression result.

\begin{theorem}\label{thm:single}
    Consider a communication network with a protocol generating $p(a \mid x)$ for $a \in A$ and $x \in X$ using at least one independent randomness source $R$.
    Given a tolerance level $\epsilon > 0$ and an integer $n$ satisfying
    \begin{equation}\label{eq:single}
        n > \frac{\ln\big( 2 \abs{X} \abs{A} \big)}{2 \epsilon^2}\ ,
    \end{equation}
    there exists a randomness source $Q$ of cardinality at most $n$ that, when sampled by the same protocol instead of $R$ and distributed to the same parties as $R$, results in a distribution $\hat{p}(a \mid x)$ with $\norm{\hat{p}(a \mid x) - p(a \mid x)}_\infty < \epsilon$.
\end{theorem}

We note that the infinity norm maximizes over all choices of $x$ and $a$.
The proof of \cref{thm:single} (in \cref{appx:proof-single}) is non-constructive -- while it gives us the procedure used to construct $Q$, it does not reveal any particular sampling yielding the required precision target.
It guarantees that at least one such sampling exists, though finding it in practice is entirely up to luck.

We can also apply this compression technique to multiple sources in succession, yielding the result below.
\begin{corollary}\label{cor:multiple}
    Consider a communication network with a protocol generating $p(a \mid x)$ for $a \in A$ and $x \in X$ using at least $m$ independent randomness sources $R_i$.
    Given a tolerance level $\epsilon > 0$, integers $n_i$, and $\delta_i > 0$ satisfying%
    \begin{equation}\label{eq:multiple}
        n_i > \frac{\ln\big( 2 \abs{X} \abs{A} \big)}{2 (\epsilon \delta_i)^2} \text{ for } i \in \{1, \ldots, m\}\ ,\ \
        \sum_{i=1}^m \delta_i = 1\ ,
    \end{equation}
    there exist randomness sources $Q_i$ of cardinalities at most $n_i$ that, when sampled by the same protocol instead of $R_i$ and distributed to the same parties as $R_i$, result in a distribution $\hat{p}(a \mid x)$ with $\norm{\hat{p}(a \mid x) - p(a \mid x)}_\infty < \epsilon$.
\end{corollary}
As mentioned above, extending \cref{thm:single} to \cref{cor:multiple} involves compressing each randomness source with a higher amount of precision.
This enables using the triangle inequality to bound the total deviation from $p(a \mid x)$ by $\epsilon$ (see \cref{appx:proof-multiple} for details).
The flexibility of exactly how to divvy up this amount is encoded in the choice of $\delta_i$ above.
In practice, this allows obtaining lower sizes on some randomness sources at the cost of larger sizes on others.
Still, when considering either the maximum or total cardinality of the randomness sources, it is often advantageous to simply take $\delta_i = 1/m$, with $m$ being the number of compressed randomness sources.
In this case, every randomness source has the same cardinality $\abs{Q_i}$ satisfying the previously stated condition in \eqref{eq:general}.
We use $\delta_i = 1/m$ throughout the rest of the Letter; as an example, applying this to the triangle scenario yields \eqref{eq:triangle}.

We additionally note that \cref{cor:multiple} preserves the structure of randomness sources under compression, by ensuring that each compressed $Q_i$ is available to the same parties as the corresponding original $R_i$.
If structure preservation is optional, then one may instead apply \cref{thm:single} on the joint randomness source $(R_1, \ldots, R_m)$ to recover a cardinality bound independent of $m$.
For instance, applying the latter approach to the triangle scenario in \cref{fig:triangle} modifies the network structure to that of a tripartite Bell scenario, where a single randomness source is shared between all parties who received any shared randomness from the original sources.

The bound in \eqref{eq:multiple} only depends on the number of the compressed randomness sources $m$, and not their original cardinality.
In practice, if a particular favorable network instance guaranteed to produce $p(a \mid x)$ is known, we may further reduce the cardinality of the compressed randomness sources by leaving some sources uncompressed.
We then get $m$ by counting only the compressed randomness sources.

\section{Comparison to existing bounds}

To the best of our knowledge, the bound in \cref{cor:multiple} is the first of its kind in bounding the dimension of shared randomness when approximating network distributions, although similar bounds exist for matching such distributions exactly.
Using Carath\'{e}odory's theorem, the size of each $R_i$ is upper bounded by $\abs{X} \abs{A} + 1$ \cite[Proposition~2]{rosset2018universal}; for $h$ parties with input and output sizes $\nx$ and $\na$ this bound is
\begin{equation}\label{eq:exact}
    (\nx \na)^h + 1\ .
\end{equation}
Note that this exact-case bound scales exponentially in the number of parties $h$, whereas our bound for approximation (with \eqref{eq:lnxa} taken into account) scales only linearly in $h$.
The approximation bound becomes exponential only if the number of randomness sources $m$ grows exponentially with $h$, which can occur when we allow a source of shared randomness between every subset of the $h$ parties, for example.
It is not known how tight the bound in \cite{rosset2018universal} is among network-agnostic bounds, but it could be the case that a similar difference remains between our bound and a tight upper bound for the exact case.
The following distribution exemplifies that there can be an exponential separation between our upper bound for approximating and a lower bound for exact matching.

\section{Examples with exponential separations}

Consider a \emph{multipartite} Bell scenario with $h$ parties sharing a single source of randomness $R$.
Suppose there are no inputs (i.e., $X_1 = \cdots = X_h = \{ 0 \}$), and output sets $A_1 = \cdots = A_h$ are of size $\na$.
Let $R$ be a random variable over $\{ 1, \ldots, \na \}$ with nonzero probabilities $q$ and suppose each party samples $r \in R$ and directly outputs the result.
This yields the distribution
\begin{equation}\label{eq:scenario-noinput}
    p(a_1, \ldots, a_h) =
    \begin{cases}
        q(a_1) & \text{if } a_1 = \cdots = a_h \\
        0 & \text{otherwise}\ .
    \end{cases}
\end{equation}
In short, they wish to generate random maximally correlated outputs.
One may ask whether this can be done with a randomness source of smaller cardinality.
But, as we show in \cref{appx:lower-noinput}, shared randomness of dimension at least $\na$ is necessary to generate this distribution exactly.
In comparison, our bound yields $(\ln(2) + h \ln(\na)) / (2 \epsilon^2)$, which is logarithmic in $\na$.

The next example has $A_1 = A_2 = \{ 0, 1 \}$, $X_1 = X_2$, and the two parties must generate the distribution
\begin{equation}\label{eq:scenario-input}
    p(a_1, a_2 \mid x_1, x_2) =
    \begin{cases}
        1/2 & \text{if } a_1 = a_2,\ x_1 = x_2\ ,\\
        0 & \text{if } a_1 \neq a_2,\ x_1 = x_2\ ,\\
        1/4 & \text{if } x_1 \neq x_2\ ,
    \end{cases}
\end{equation}
meaning their outputs must be equiprobable and matching when their inputs match, and uniformly random otherwise.
One valid strategy that attains this distribution exactly is for Alice and Bob to flip $\log_2(\nx+1)$ shared coins (i.e., a randomness source of size $\nx+1$) and output the XOR of a subset determined by their respective inputs.
Moreover, the amount of shared randomness used by this strategy matches the minimum necessary for this scenario; see \cref{appx:lower-input} for a lower bound.
The exact upper bound attains similarly polynomial scaling with $(2 \nx)^2 + 1$, yet the approximation bound yields $e^{-2} (\ln(8)/2 + \ln(\nx))$, which is logarithmic in $\nx$.

We note one additional consequence of the exponential separation between the exact and approximate case upper bounds.
It may seem that the premise of \cref{cor:multiple} makes it ignore much of the relevant network structure that could otherwise be utilized for a tighter bound.
Though this may well be the case, here we illustrate that the scaling of our bound makes such a difference negligible.
Suppose Network~1 uses $r_1$ to generate $p_1(a_1 \mid x_1) = \E_{r_1} \big[ p_1(a_1 \mid x_1, r_1) \big]$ and Network~2 uses $r_2$ to generate $p_2(a_2 \mid x_2) = \E_{r_2} \big[ p_2(a_2 \mid x_2, r_2) \big]$, operating in conjunction as in \cref{fig:product}.
Their joint distribution is $p(a_1, a_2 \mid x_1, x_2) = p_1(a_1 \mid x_1) \, p_2(a_2 \mid x_2) = \E_{r_1, r_2} \big[ p_1(a_1 \mid x_1, r_1) \, p_2(a_2 \mid x_2, r_2) \big]$.

\begin{figure}[ht]
    \centering
    \begin{tikzpicture}[> = stealth, semithick]
    \tikzstyle{randomness}=[circle,draw=black,fill=black!20,inner sep=0.5mm]
    \tikzstyle{input}=[draw=black,fill=black!5,inner sep=1.5mm]
    \tikzstyle{output}=[draw=black,fill=black!5,inner sep=1.5mm]
    \tikzstyle{party}=[thick,draw=black,fill=white,inner sep=1.5mm]
    \node[party] (N1) at (-1.5,0) {Network 1};
    \node[party] (N2) at (1.5,0) {Network 2};
    \draw[dashed] (0,-1.25) -- (0,1.25);
    \node[output,below=.5cm of N1] (a1) {$a_1$};
    \node[output,below=.5cm of N2] (a2) {$a_2$};
    \draw[->] (N1) -- (a1);
    \draw[->] (N2) -- (a2);
    \node[randomness,above=.5cm of N1,xshift=.5cm] (r1) {$r_1$};
    \node[randomness,above=.5cm of N2,xshift=.5cm] (r2) {$r_2$};
    \draw[~] (r1) -- (N1);
    \draw[~] (r2) -- (N2);
    \node[input,above=.5cm of N1,xshift=-.5cm] (x1) {$x_1$};
    \node[input,above=.5cm of N2,xshift=-.5cm] (x2) {$x_2$};
    \draw[->] (x1) -- (N1);
    \draw[->] (x2) -- (N2);
\end{tikzpicture}
    \caption{Product network composition.}
    \label{fig:product}
\end{figure}
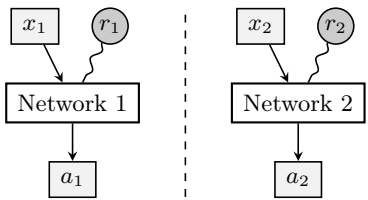

Suppose we are aware that the networks are completely separated, and thus attempt to bound their composition by combining individual bounds.
When the respective cardinalities of the inputs and outputs are equal (i.e., $\abs{X_1} = \abs{X_2} = \nx$ and $\abs{A_1} = \abs{A_2} = \na$), the size of the randomness sources $\nx \na + 1$ does not change in the exact case.
In the approximate case we need to bound each network separately to $\epsilon/2$ precision, yielding a combined bound of $\frac{1}{2} \left( \frac{2}{\epsilon} \right)^2 [\ln(2) + \ln(\nx \na)]$.
But, what if instead we are unaware of the separation and compress the product network directly?
Now, the exact bound is $(\nx \na)^2 + 1$, whereas the approximate bound is $\frac{1}{2} \left( \frac{2}{\epsilon} \right)^2 [\ln(2) + 2 \ln(\nx \na)]$.
Note how these compare to the bounds that account for the separation; while the exact bound is roughly the product of its components, the approximate bound only doubles.
Any ignored network structure thus does not greatly contribute to significant differences in the upper bound for the approximate case.

Our approximate case bound attains exponentially better scaling not only in the network size $h$ but also in the size of the question and answer sets.
Fixing $h$ and supposing that each of the $h$ parties has input and output sets of cardinality $\nx$ and $\na$ respectively, the approximate case bound in \eqref{eq:general} scales logarithmically with respect to $\nx \na$ whereas the exact case bound in \eqref{eq:exact} scales polynomially with degree $h$.

In spite of the improvements in asymptotic scaling, it is reasonable to wonder whether the various constants attached to the approximate case bound make it viable for networks of practical sizes.
To this end, we provide numerical comparisons of the bounds in \cref{fig:n,fig:epsilon}.

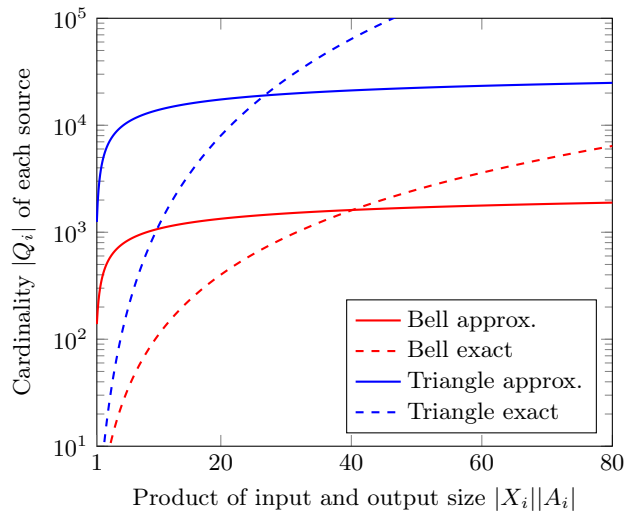
\begin{figure}[htp]
    \centering
    \begin{tikzpicture}
    \begin{axis}[
        ymode=log,
        legend pos=south east,
        legend cell align={left},
        every axis plot/.append style={thick},
        xmin=1, xmax=80,
        ymin=1e1, ymax=1e5,
        domain=1:80,
        samples at={1,1.125,...,3,3.25,3.5,...,7,7.5,8,...,15,16,17,...,32,34,36,...,64,68,72,76,80},
        extra x ticks={1},
        xlabel={Product of input and output size $\nx \na$},
        ylabel={Cardinality $\abs{Q_i}$ of each source},
    ]
        \newcommand{\fAppx}{(\varM/0.05)^2*ln(2*x^\varH)/2};
        \newcommand{\fExact}{x^\varH+1};
        \newcommand{\varH}{2};\newcommand{\varM}{1};
        \addplot[color=red]{\fAppx};
        \addplot[color=red,style=dashed]{\fExact};
        \renewcommand{\varH}{3};\renewcommand{\varM}{3};
        \addplot[color=blue]{\fAppx};
        \addplot[color=blue,style=dashed]{\fExact};
        \legend{Bell approx.,Bell exact,Triangle approx.,Triangle exact};
    \end{axis}
\end{tikzpicture}
    \caption{%
        Comparison of the exact upper bound in \eqref{eq:exact} and our approximate upper bound in \eqref{eq:general} on the cardinality of each source of randomness, plotted as a function of the size of each party's inputs and outputs.
        Here we have fixed the error tolerance to be $\epsilon = 0.05$.
    }
    \label{fig:n}
\end{figure}

In \cref{fig:n} we plot the bounds directly, comparing the cardinality $\abs{Q_i}$ of each source of shared randomness attained by the two bounds.
Specifically, picking any set of parameters above the plotted lines guarantees that any distribution can be constructed (possibly $\epsilon$-approximately) for a network of a given size with said amount of shared randomness, as guaranteed by the respective bound.
As such, for any values of input/output sizes $\nx \na$ where the solid lines lie below the dashed lines of the respective color, our approximate case bound outperforms the exact case bound.

In addition to the approximate case bound performing comparatively better for larger values of $\nx \na$, the same is true for larger values of $h$ -- as evidenced in \cref{fig:n} by the blue lines in the triangle scenario intersecting at a lower value of $\nx \na$ compared to the red lines for the Bell scenario.
If we instead focus solely on such points of intersection and how they vary with other parameters, we obtain \cref{fig:epsilon}.
Areas above the plotted lines indicate instances where our approximate case bound outperforms the existing exact case bound.
This performance improvement is more significant for networks with large $h$ and small $m$.

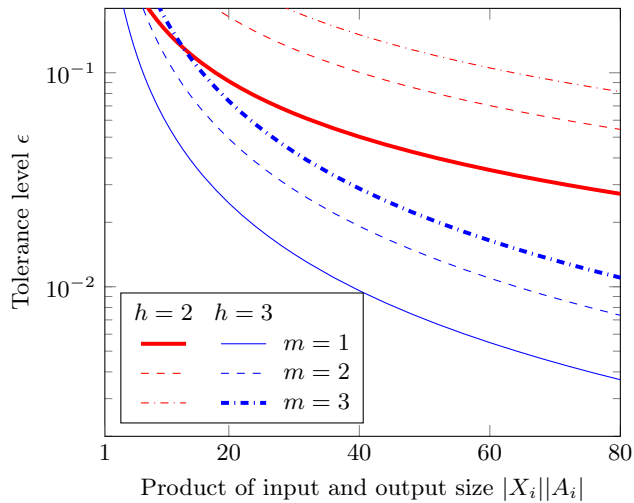
\begin{figure}[htp]
    \centering
    \begin{tikzpicture}
    \begin{axis}[
        ymode=log,
        legend pos=south west,
        xmin=1, xmax=80,
        ymin=2e-3, ymax=2e-1,
        domain=1:80,
        samples at={1,1.5,...,8,9,10,...,16,18,20,...,32,36,40,...,80},
        extra x ticks={1},
        xlabel={Product of input and output size $\nx \na$},
        ylabel={Tolerance level $\epsilon$},
        legend columns=2,
    ]
        \newcommand{\fEqual}{\varM/(2*(x^\varH+1)/ln(2*x^\varH))^(1/2)};
        \addlegendimage{legend image with text={$h=2$}};
        \addlegendentry{};
        \addlegendimage{legend image with text={$h=3$}};
        \addlegendentry{};
        \newcommand{\varH}{2};\newcommand{\varM}{1};
        \addplot[color=red,line width=1.5pt]{\fEqual};
        \addlegendentry{};
        \renewcommand{\varH}{3};
        \addplot[color=blue]{\fEqual};
        \addlegendentry{$m=1$};
        \renewcommand{\varH}{2};\renewcommand{\varM}{2};
        \addplot[color=red,style=dashed]{\fEqual};
        \addlegendentry{};
        \renewcommand{\varH}{3};
        \addplot[color=blue,style=dashed]{\fEqual};
        \addlegendentry{$m=2$};
        \renewcommand{\varH}{2};\renewcommand{\varM}{3};
        \addplot[color=red,style=dash dot]{\fEqual};
        \addlegendentry{};
        \renewcommand{\varH}{3};
        \addplot[color=blue,style={dash dot},line width=1.5pt]{\fEqual};
        \addlegendentry{$m=3$};
    \end{axis}
\end{tikzpicture}
    \caption{%
        Minimum value of tolerance level $\epsilon$ such that the approximate bound in \eqref{eq:general} is tighter than the exact bound in \eqref{eq:exact}.
        Thicker lines correspond to the Bell and triangle scenarios.
    }
    \label{fig:epsilon}
\end{figure}

\section{Conclusions and future work}

In this Letter, we demonstrate how to compress the size of shared randomness in network settings.
We derived a novel bound on the amount of shared randomness needed to approximate any probability distribution generated by a communication network, which scales much better than the best currently-known upper bound when matching the distribution exactly.
We also provided exponential separations between these cases by providing our upper bound as well as proving an exact lower bound for specific distributions.

A potential future direction is to examine the possibility of compressing other types of resources in networks.
As large quantum resources are typically more difficult to create, maintain, and utilize in comparison to classical resources, compressing them is of great practical interest.
This is studied in prior work \cite{stark2016compressibility,panahi2025upper}, which approximate distributions of quantum models and bound local Hilbert space dimension.
Recent work \cite{coladangelo2020two} additionally provides strong evidence against the possibility of compressing shared entanglement, as there exist non-local games that may only be approximated $\epsilon$-close to their value only with entanglement of dimension $2^{\Omega(\epsilon^{-1/8})}$.
Nevertheless, confirming the impossibility of quantum resource compression would likewise be interesting in providing a distinction with the efficient compression of shared randomness shown in this Letter.

\section{Acknowledgments}

Y.U.~was supported by the Perimeter Institute for Theoretical Physics and funded by Mike and Laura Serbinis.
Part of the work was done when Y.U., J.S., and A.A.~were at the Perimeter Institute.
A.A.~acknowledges support through the NSF award QCIS-FF: Quantum Computing \& Information Science Faculty Fellow at Harvard University (NSF 2013303), and NSF Award Numbers 2430375, 2238836.
J.S.~and A.Z.~are supported in part by the Commonwealth of Virginia’s Commonwealth Cyber Initiative (CCI) under grant number 469351.

\nocite{*}
\bibliography{references}

\appendix

\section{Proof of Theorem~\ref{thm:single}}\label{appx:proof-single}

\begin{proof*}
    Let $r^{(1)}, \ldots, r^{(n)}$ be random variables representing $n$ independent samples from $R$ according to its probabilities.
    Define $Q$ to be the uniform distribution over these samples (proportional to repeated samples), meaning
    \begin{equation}
        Q(r) = \frac{1}{n} \sum_{i=1}^n \mathbbm{1}[ r^{(i)} = r]\ ,
    \end{equation}
    where $\mathbbm{1}$ is the indicator function, and hence $\abs{Q} \leq n$.
    We consider the following protocol: all parties follow their previous strategies, sampling from the same shared resources -- save for $R$, for which a sample $r$ is instead drawn from $Q$ and distributed to the same parties as previously.
    This produces the distribution
    \begin{equation}
        \hat{p}(a \mid x) = \frac{1}{n} \sum_{i=1}^n p(a \mid x, r^{(i)})\ ,
    \end{equation}
    which approximates $p(a \mid x)$ since $\E \big[ \hat{p}(a \mid x) \big] = p(a \mid x)$.
    Fix an arbitrary $(a,x) \in A \times X$.
    Then, $p( a \mid x, r^{(i)})$ for each $i \in \{1, \ldots, n\}$ is a collection of $n$ independent random variables bounded between $0$ and $1$.
    By Hoeffding's inequality \cite{chernoff1952measure,hoeffding1963probability},
    \begin{equation}
        \Pr \Big( \abs[\big]{p(a \mid x) - \hat{p}(a \mid x)} \geq \epsilon \Big) \leq 2 \exp(-2 n \epsilon^2)\ .
    \end{equation}
    By the union bound,
    \begin{equation}\label{eq:multi-union}
    \begin{aligned}
        \Pr \Big( \exists (a,x) \in A \times X, \abs[\big]{p(a \mid x) - \hat{p}(a \mid x)} \geq \epsilon \Big) \\
        \leq 2 \abs{A} \abs{X} \exp(-2 n \epsilon^2)\ ,
    \end{aligned}
    \end{equation}
    which once inverted becomes
    \begin{equation}
    \begin{aligned}\label{eq:proof-single-prob}
        \Pr \Big( \forall (a,x) \in A \times X, \abs[\big]{p(a \mid x) - \hat{p}(a \mid x)} < \epsilon \Big) \\
        > 1 - 2 \abs{A} \abs{X} \exp(-2 n \epsilon^2)\ .
    \end{aligned}
    \end{equation}
    Thus when we have
    \begin{equation}\label{eq:proof-single-condition}
        1 - 2 \abs{A} \abs{X} \exp(-2 n \epsilon^2) \geq 0\ ,
    \end{equation}
    the left-hand probability of \eqref{eq:proof-single-prob} is nonzero, meaning there exists a sampling $r^{(1)}, \ldots, r^{(m)}$ such that
    \begin{equation}
        \norm{p(a \mid x) - \hat{p}(a \mid x)}_\infty = \hspace*{-1em} \max_{\raisebox{-0.25em}{$\scriptstyle (a,x) \in A \times X$}} \hspace*{-0.8em} \abs{p(a \mid x) - \hat{p}(a \mid x)} < \epsilon\ .
    \end{equation}
    Rearranging \eqref{eq:proof-single-condition} completes the proof.
\end{proof*}

\section{Proof of Corollary~\ref{cor:multiple}}\label{appx:proof-multiple}

\begin{proof*}
    Denote $p_0(a \mid x) = p(a \mid x)$, and for all $i \in \{1, \ldots, m\}$ let $p_i(a \mid x)$ be the correlations of the network obtained by compressing $R_i$ up to tolerance $\epsilon \delta_i$ compared to $p_{i-1}(a \mid x)$.%
    \begin{equation}
    \begin{aligned}
        \norm{\hat{p}(a \mid x) - p(a \mid x)}_\infty
        &\leq \sum_{i=1}^m \norm{p_i(a \mid x) - p_{i-1}(a \mid x)}_\infty\hspace*{-.25em} \\
        &\leq \sum_{i=1}^m \epsilon \delta_i
        = \epsilon\ .
    \end{aligned}
    \end{equation}
    The conditions for this are given by \cref{thm:single} as
    \begin{equation}
        n_i \geq \frac{\ln\big( 2 \abs{A} \abs{X} \big)}{2 (\epsilon \delta_i)^2} \quad
        \forall i \in \{1, \ldots, m\}\ ,
    \end{equation}
    as needed.
\end{proof*}

\section{Lower bound on scenario (\ref{eq:scenario-noinput})}\label{appx:lower-noinput}

We consider an earlier example of generating correlations in a multipartite Bell scenario with no inputs.
There are $h \geq 2$ parties with no communication except sharing the sampled value $r$ of a randomness source $R$, and they wish to generate perfectly correlated outputs with their individual outputs modeling some distribution $q$.
We show that to do so, the cardinality of the shared randomness source needs to be at least $\na := \abs{\supp(q)}$.

Under this setup, each possible generated distribution $p(a_1, \ldots, a_n)$ is of the form
\begin{equation}
     \sum_{r \in \supp(q)} p(r)\, p_{A_1}(a_1 \mid r) \cdots\, p_{A_h}(a_h \mid r)\ ,
\end{equation}
with $p_{A_i} (a_i \mid r)$ being the probability of the $i$th party outputting $a_i$ when receiving $r$.
Assuming that the desired perfectly correlated output $\mathbbm{1}[a_1 = \ldots = a_h] \, q(a_1)$ is attained, we first show that any individual player's strategy must be deterministic.

\begin{lemma}
    The strategy $p_{A_1}$ is deterministic.
    That is, for any $r \in \supp(p) = \{ r : p(r) > 0 \}$ there exists an $a \in A_1$ such that $p_{A_1} (a \mid r) = 1$.
\end{lemma}
\begin{proof}
    Suppose an arbitrary $r \in \supp(p)$ is sampled and a second party outputs $b \in A_2$.
    To produce the desired distribution, the first party must then also output $b$ with probability $1$ when given the same $r$.
\end{proof}

As such, all randomness in this setting comes from the shared source.
Thus, each value sampled from the shared randomness source maps to only one output, meaning its cardinality must be at least the number of outputs.
We present this argument more formally below.

\begin{lemma}\!\!
    For the scenario in \eqref{eq:scenario-noinput}, $\abs{R}\!\geq\!\abs{\supp(q)}$ is necessary.
    Moreover, $R = \supp(q)$ suffices.
\end{lemma}
\begin{proof}
    The ``suffices'' part is trivial by each party outputting $r$.
    For every $a \in A_1$, we define
    \begin{equation}
        Q_a \coloneqq \{r : p(r) > 0,\ p_{A_1}(a \mid r) = 1 \}\ .
    \end{equation}
    We now examine $p(a_1)$, i.e., the marginal of $A_1$.
    \begin{equation}
    \begin{aligned}
        p(a_1) &= \sum_{\substack{a_2, \ldots, a_h \\ r \in \supp(q)}} p(r)\, p_{A_1} (a_1 \mid r) \cdots\, p_{A_h} (a_h \mid r) \\
        &= \sum_{r \in \supp(q)} p(r)\, p_{A_1} (a_1 \mid r)
        = \sum_{r \in Q_{a_1}} p(r)\ .
    \end{aligned}
    \end{equation}
    And for the desired distribution,
    \begin{equation}
        p(a_1) = \sum_{a_2,\ldots,a_h} \mathbbm{1}[a_1=\ldots=a_h]\, q(a_1) = q(a_1)\ .
    \end{equation}
    Thus, $q(a) = \sum_{r \in Q_a} p(r)$.
    For every $a$ such that $q(a) > 0$ we require $Q_a \neq \varnothing$, meaning $\abs{\supp(q)}$ is the number of nonempty $Q_a$s.
    Since the $Q_a$s have empty intersections, we have $\abs{R} \geq \sum_{a \in A_1} \abs{Q_a} \geq \abs{\supp(q)}$, as needed.
\end{proof}

\section{Lower bound on scenario (\ref{eq:scenario-input})}\label{appx:lower-input}

We consider an earlier example where two parties receive one of $\nx$ inputs and must generate bit outputs that are correlated precisely when said inputs match.
We show that to do so, the cardinality of the shared randomness source needs to be at least $\nx + 1$.

Similar to the previous example, the generated distribution $p(a_1, a_2 \mid x_1, x_2)$ is of the form
\begin{equation}
     \sum_{r \in R} p(r)\, p_{A}(a_1 \mid x_1, r)\, p_{B}(a_2 \mid x_2, r)\ ,
\end{equation}
and we begin by ruling out the use of private randomness.

\begin{lemma}\label{lem:scenario-input-det}
    The strategy $p_{A}$ is deterministic.
    That is, for any $x \in X_1$ and $r \in \supp(p) = \{ r : p(r) > 0 \}$ there exists an $a \in A_1$ such that $p_{A} (a \mid x, r) = 1$.
\end{lemma}
\begin{proof}
    This comes from the outputs needing to be correlated for matching inputs.
    Fix $x \in X_1$ and $r \in \supp(p)$, and suppose the second party outputs $b \in A_2$.
    To produce the desired distribution, the first party must then also output $b$ with probability $1$ for the same $x$ and $r$.
\end{proof}

We then use an orthogonality argument to lower bound the cardinality of shared randomness.

\begin{lemma}
    For the scenario in \eqref{eq:scenario-input}, $\abs{R} \geq \nx + 1$ is necessary.
    Moreover, $\abs{R} = 2^{\ceil{\log_2(\nx + 1)}}$ suffices.
\end{lemma}

\begin{proof}
    For the ``suffices'' part, let the shared randomness be $k = \ceil{\log_2(\nx+1)}$ independent and uniform bits $b_1, \ldots, b_k$.
    We then have that
    \begin{equation}
        \Big\{ \bigoplus_{i \in S} b_i \Big\}_{S \in \mathcal{P}([k]) \setminus \varnothing}\ ,
    \end{equation}
    where $\mathcal{P}([k])$ is the power set of $\{1, \ldots, k\}$, is a collection of $2^k-1$ pairwise independent and uniform RVs.
    The two parties can then sample the variable corresponding to their (one of $\nx \leq 2^k-1$) input values.

    Without loss of generality, consider the distribution $(p(r_1), \ldots, p(r_m))$ with $r_i \in \supp(p)$ of the shared randomness source.
    We have $p(r_i) > 0$ and $\sum_i p(r_i) = 1$.
    Define a weighted inner product on $\mathbb{C}^m$ via
    \begin{equation}
        \inner{x}{y}_p \coloneqq \sum_{i=1}^m p(r_i) \, x_i \, \overline{y_i}\ ,
    \end{equation}
    which is a valid inner product.
    Since strategies are deterministic by \cref{lem:scenario-input-det}, for each $x \in X_i$ we represent the two parties' strategies by vectors $u_x \in \{0,1\}^m$ defined as
    \begin{equation}
        (u_x)_i \coloneqq p_A(1 \mid x, r_i) = p_B(1 \mid x, r_i)\ ,
    \end{equation}
    where equivalence follows from the outputs needing to be perfectly correlated when receiving the same $x$.
    That is, the $i$th element is $1$ when $r_i$ is sampled and the corresponding party outputs $1$, and $0$ otherwise.
    We have
    \begin{equation}
        p(1,1 \mid x,y) = \inner{u_x}{u_y}_p = \begin{cases} 1/2 & x = y\ ,\\ 1/4 & x \neq y\ .\end{cases}
    \end{equation}
    Now, define $v_x = u_x - \mathbf{1}/2 \in \{1/2, -1/2\}^m$, where $\mathbf{1}$ is the all-ones vector.
    This centers the vectors (that is, $\inner{v_x}{\mathbf{1}}_p = 0$), and we get $\inner{v_x}{v_y}_p = \mathbbm{1}[x=y]/4$ since $\inner{u_x}{\mathbf{1}}_p = p_A(1 \mid x) = 1/2$.
    For simplicity, let $n \coloneqq \abs{X_i}$.
    The vectors $v_1, \ldots, v_{n}$ are thus pairwise orthogonal and nonzero, and hence linearly independent.
    In particular, $c_1 v_1 + \ldots + c_{n} v_{n} = 0$ implies
    \begin{equation}
        c_1 \inner{v_1}{v_j}_p + \ldots + c_{n} \inner{v_{n}}{v_j}_p = \inner{0}{v_j}_p\ ,
    \end{equation}
    then $c_j \inner{v_j}{v_j}_p = 0$ and $c_j = 0$ for all $j \in \{1, \ldots, n\}$.
    The vectors are in the subspace
    \begin{equation}
        \mathbf{1}^\perp \coloneqq \{ v \in \mathbb{R}^m : \inner{v}{\mathbf{1}}_p = 0 \}
    \end{equation}
    of dimension $m-1$, which corresponds to an upper bound on $n$ since $v_1, \ldots, v_n$ are linearly independent.
    We thus have $\abs{R} \geq \abs{\supp(p)} = m \geq n + 1 = \nx + 1$.
\end{proof}

\end{document}